\documentclass[a4paper,UKenglish]{lipics-v2016}
\usepackage[utf8]{inputenc}
\usepackage{microtype}
\usepackage{amsthm,amsmath,amssymb}
\usepackage{xargs,calc}
\usepackage{stmaryrd}
\usepackage{xcolor}
\usepackage{graphicx}

\usepackage{caption}

\usepackage{hyperref}

\theoremstyle{plain}
\newtheorem*{theorem*}{Theorem}
\newtheorem{notation}[theorem]{Notation}
\newtheorem*{notation*}{Notation}

\newtheorem{property}[theorem]{Property}
\newtheorem{proposition}[theorem]{Proposition}
\newtheorem{numremark}[theorem]{Remark}

\newcommand{\AutomatonScale}{0.6}

\title{An efficient algorithm to decide periodicity of $b$-recognisable sets using MSDF convention}
\author{Bernard Boigelot}
\author{Isabelle Mainz}
\author{Victor Marsault\thanks{Supported by a Marie Sk{\l}odowska-Curie fellowship,
co-funded by the European Union.}}
\author{Michel Rigo}
\affil{Montefiore Institute \& Department of Mathematics, Universit\'e de Li\`ege, Belgium \\
          \texttt{\{bernard.boigelot, isabelle.mainz, victor.marsault, m.rigo\}@ulg.ac.be}}

\authorrunning{B. Boigelot, I. Mainz, V. Marsault, and M. Rigo}
\titlerunning{An efficient algorithm to decide periodicity of $b$-recognisable sets with MSDF}

\keywords{integer-base systems; automata; recognisable sets; periodic sets.}
\DOIPrefix{}
  \newcommand{\ZZ}[1][p]{\mathbb{Z}/#1\mathbb{Z}}
\newcommandx{\ARpb}[3][1=R,2=p,3=b]{\mathcal{A}_{(#2,#1)}}
\newcommandx{\BIb}[2][1=I,2=b]{\mathcal{B}_{#1}}
\newcommandx{\CRpIb}[4][1=R,2=p,3=I,4=b]{\mathcal{C}_{(#2,#1,#3)}}
\newcommandx{\intalph}[1][1={b}]{\llbracket #1 \rrbracket}
\newcommandx{\val}[2][2={}]{\overline{\,#1\,}^{#2}}
\newcommandx{\rep}[2][2={}]{\langle#1\rangle_{#2}}
\newcommand{\wlen}[1]{|{#1}|}
\newcommand{\Ac}{\mathcal{A}}
\newcommand{\Bc}{\mathcal{B}}
\newcommand{\Mc}{\mathcal{M}}
\newcommand{\Gc}{\mathcal{G}}
\newcommand{\N}{\mathbb{N}}
\newcommand{\crl}[1]{\langle#1\rangle}
\renewcommand{\mod}{{\sf\%}}
\newsavebox{\vmcommentbox}
  {%
    \newcommand\colboxcolor{#1}%
    \begin{lrbox}{\vmcommentbox}%
      \begin{minipage}{\textwidth-4\parindent}%
        \stepcounter{thm}%
        \textbf{Commentaire \thethm{ }(VM).} \itshape %
  }{%
      \end{minipage}%
    \end{lrbox}%
    \par\noindent\hfill%
    \colorbox{\colboxcolor}{\usebox{\vmcommentbox}}
    \hfill\hspace*{0cm}%
  }

\newcounter{subthm}\newcounter{cond}

\newcommand{\sublabel}[1]{\refstepcounter{cond}\label{#1*}\refstepcounter{subthm}\label{#1}}

\newcommand*{\minwidthbox}[2]{%
  \makebox[{\ifdim#2<\width\width\else#2\fi}][c]{#1}%
}
\newcommand{\rpath}[1]{
\xrightarrow{\minwidthbox{$\scriptstyle\hspace{0.4ex}#1\hspace{0.6ex}$}{2.75ex}}}

\makeatletter
\newcommand{\vm@date@separator}{\hspace*{0.15ex}\rule[0.4\vm@date@height]{1ex}{0.07\vm@date@height}\hspace*{0.15ex}}
\newcommand{\vmdatefont}[1]{#1}
\newcommand{\isotoday}{%
  \vmdatefont{
    \newdimen\vm@date@height%
    \setbox0=\hbox{0123456789}%
    \vm@date@height=\ht0 \advance\vm@date@height by -\dp0
    \the\year\vm@date@separator\two@digits{\month}\vm@date@separator\two@digits{\day}%
  }%
}
\makeatother

\begin{document}

  \maketitle

  \begin{abstract}

  Given an integer base $b>1$, a set of integers is represented in base~$b$ by a
  language over~$\{0,1,...,b-1\}$.  The set is said to be $b$-recognisable if its
  representation is a regular language.  It is known that eventually periodic sets
  are $b$-recognisable in every base $b$, and Cobham's theorem implies the converse:
  no other set is $b$-recognisable in every base $b$.

  We are interested in deciding whether a $b$-recognisable set of integers
  (given as a finite automaton) is eventually periodic.  Honkala showed that this
  problem decidable in 1986 and recent developments give efficient decision
  algorithms.  However, they only work when the integers are written with the
  least significant digit first.

  In this work, we consider the natural order of digits (Most Significant
  Digit First) and give a quasi-linear algorithm to solve the problem in this
  case.

  \end{abstract}

%

\section*{Introduction}

Let $b>1$ be an integer base. We let $\intalph = \{ 0,1,\ldots, b-1\}$ denote the canonical alphabet of base-$b$ digits. If $u=u_\ell\cdots u_0$ belongs to $\intalph^*$, we let $\val{u}$ denote the \emph{value} of~$u$ in base~$b$, \emph{i.e.}, $\val{u}=\sum_{i=0}^\ell u_i\, b^i$. Note that the leftmost digit is the most significant one. We let $\rep{n}$ denote the (shortest) \emph{base-$b$ representation} of $n$. We set $\rep{0}$ to be the empty word $\varepsilon$. If reference to the base $b$ is needed, we write $\rep{n}[b]$. Thus $\rep{n}$ is the unique word $u$ over $\intalph$ not starting with $0$ and such that $\val{u}=n$. Moreover, for every $u\in\intalph^*$ such that $\val{u}=n$, there exists $i\ge 0$ such that $u=0^i\rep{n}$.

\subsection*{Our contribution}

Let $b>1$ be an integer base. In this paper, we develop an algorithm to decide whether a given deterministic automaton~$\Ac$ over the alphabet $\intalph$ accepts, by value, an (eventually) periodic set of integers. More precisely, the question is to decide whether there exist integers $p\ge 1$ and $N\ge 0$ such that,
for all words $u\in\intalph^*$, if $\val{u}\ge N$, then $u$ is accepted by $\Ac$ if and only if $\rep{\val{u}+p}$ is accepted as well. \emph{Acceptance by value} means that words sharing the same value are either all accepted or all rejected. Stated otherwise, a word $u$ is accepted by $\Ac$ if and if only if $0u$ is accepted. The main result of this paper is the following one.

\begin{theorem}\label{t.main}
  Given an integer base~$b>1$ and a~$n$-state deterministic automaton~$\Ac$ over the
  alphabet~{$\intalph$}, it is decidable in~$O(b n \log n)$ time
  whether or not~$\Ac$ accepts, by value, some eventually periodic set of integers.
\end{theorem}

We stress the fact that the input automaton $\Ac$ reads words most significant digit first (MSDF). This is an important difference with other results discussed in the literature. For instance, an efficient algorithm to solve this decision problem is provided for automata reading least significant digits first (LSDF) \cite{MarsaultSakarovitch2013}. One can therefore think that it is enough to take the reversal of $\Ac$ and thus consider entries LSDF. Nevertheless, the reversal of $\Ac$ has first to be determinised. This potentially leads to an exponential blow-up in the number of states and thus to an inefficient procedure.

\subsection*{Motivations and related results}

We say that a set $X\subseteq\N$ is \emph{$b$-recognisable} if $\rep{X}[b]$ is accepted by some finite automaton. One reason why eventually periodic sets of integers play a special role comes from the celebrated theorem of Cobham about the dependence to the base of $b$-recognisability.

\begin{theorem*}[Cobham, \cite{Cobham1969}]
Let $b,c>1$ be two multiplicatively independent integers. A set $X$ of integers is such that the languages $\rep{X}[b]$ and $\rep{X}[c]$ are both accepted by finite automata if and only if $X$ is eventually periodic.
\end{theorem*}

In combinatorics on words, when studying morphic words (for details and definitions, for instance, see \cite{AlloucheShallit2003,BertRigo10-b}), Cobham's theorem can be reformulated as follows. Let $b,c>1$ be two multiplicatively independent integers. An infinite word $\mathbf{x}$ is both $b$-automatic and $c$-automatic if and only if $\mathbf{x}$ is of the form $uv^\omega$ where $u,v$ are finite words. Indeed, a set of integers is $b$-recognisable if and only if its characteristic sequence is $b$-automatic. The decision problem considered in our Theorem~\ref{t.main} is well known to be decidable.

\begin{theorem*}[Honkala, \cite{Honkala1986}]
It is decidable whether or not a given $b$-automatic word is ultimately periodic.
\end{theorem*}

Complexity issues are however not all considered in Honkala's paper. The decidability of our problem of interest can also be obtained using a first-order logic characterization of $b$-recognisable sets given by B\"uchi's theorem, and the fact that Presburger arithmetic is decidable \cite{BruyereHanselMichauxVillemaire,AlloucheRampersadShallit2009}. These independent approaches all lead to decision procedures with exponential complexity.

Using LSDF convention, efficient decision procedures are known. First, Leroux obtained a quadratic decision procedure \cite{Leroux2005} for utimately-periodic $b$-recognisable sets of integers. Then, the result was improved as follows.

\begin{theorem*}[Marsault, Sakarovitch, \cite{MarsaultSakarovitch2013}]
Given an integer base~$b>1$ and a~$n$-state deterministic automaton~$\Ac$ over the
  alphabet~{$\intalph$}, it is decidable in~$O(b\, n\log n)$ time
  whether or not~$\Ac$ accepts, with LSDF convention, some eventually periodic set of integers.
\end{theorem*}

Leroux's result is stated in a multi-dimensional setting, \emph{i.e.}, the problem is to decide whether or not a $b$-recognisable subset of $\N^d$ is semi-linear. In that direction, see \cite{Semenov1977,Muchnik2003,Leroux2005}.



\subsection*{Generalisation to real numbers}

Real numbers can be encoded in a base $b > 1$ by extending
positional encoding to infinite words: A word encoding a real is
composed of a finite prefix corresponding to an integer part, followed
by a single occurrence of a distinguished symbol acting as a
separator, and an infinite suffix representing a fractional
part. Infinite-word automata are then able to recognise sets of
reals. It has been established that \emph{weak deterministic
  automata}, a restricted class of infinite-word automata, are
sufficiently expressive for recognising all sets definable in mixed
integer and real first-order additive
arithmetic~\cite{BoigelotJodogneWolper2005}.

The properties of sets of real numbers that can be recognised by weak
deterministic automata in all bases $b > 1$ have been
investigated~\cite{BoigelotBrusten2009}. Such sets generalise to the
real domain the notion of eventual periodicity; they precisely
correspond to finite combinations of eventually periodic sets of
integers, and intervals of $[0, 1]$. Checking whether an automaton
recognises such a set can be done by first splitting this automaton
into finite-state machines operating on the integer and fractional
parts of encodings. The former are then checked in the same way as
for MSDF integer encodings, and the latter by verifying that they
obey the simple structure documented in~\cite{BoigelotBrusten2009},
which is a simple operation.
As a consequence, the algorithm developed
in this paper also leads to an efficient procedure for checking that
a weak deterministic automaton recognises an eventually periodic set
of reals.

\subsection*{Generalisation to other numeration systems}

Automatic words form a particular class of morphic words. Similarly, integer-base systems are special cases of more general numeration systems such as those built on a linear recurrent sequence. One can define a \emph{numeration system} as a one-to-one map $s$ from $\N$ to a language $L$ over a finite alphabet. The integer $n$ is mapped to its representation $s(n)$ within the considered system. Hence, it is natural to ask, for given a numeration system $s$ and a subset $M$ of $L$ accepted by a finite automaton $\Ac$, whether or not the $s$-recognisable set $s^{-1}(M)\subseteq\N$ is eventually periodic.

On the one hand, Honkala's result is extended as follows. It is decidable whether or not a given morphic word is ultimately periodic \cite{Durand2013,Mitrofanov2011}. On the other hand, B\"uchi's theorem can be extended to linear numeration systems whose characteristic polynomial is the minimal polynomial of a Pisot number. See, for details, \cite{BruyereHansel1995}. In that setting, several decision problems in combinatorics on words, including the ultimate periodicity problem, are decidable \cite{CharlierRampersadShallit2012}. Using Honkala's techniques, the decision problem considered in our Theorem~\ref{t.main} is generalized to a large class of numeration systems in \cite{BellCharlierFraenkelRigo2009}. In particular, there are systems in this class for which the logical setting may not be applied. For all these decidability results presented in a wider context, no efficient procedure is known.

  \section{Preliminaries}


In this paper, we only consider deterministic accessible finite automata with
an input alphabet of the form $\intalph$. We use the acceptance-by-value convention. Thus, we may assume that the initial state bears a loop with label $0$. In particular, this will always be the case after minimisation.
Let $\Ac$ be an automaton. Its
set of states (resp. its initial state, its set of final states) is
denoted by $Q_\Ac$ (resp.~$i_\Ac$, $F_\Ac$). If the considered
automaton is clear from the context, $(s\cdot u)$ is the state~$s'$ such
that~$s\rpath{u}s'$.  The language accepted by $\Ac$ is denoted by
$L(\Ac)$. In this section, we recap basic results about automata.

\subsection{Automaton morphisms and pseudo-morphisms}

\begin{definition}\label{d.morp}
  Given two (accessible) automata~$\Ac$ and~$\Mc$ over $\intalph$, an \emph{automaton morphism} $\Ac\rightarrow\Mc$
  is a function~$\phi:Q_\Ac\to Q_\Mc$ that satisfies:
    \begin{gather}
    \label{eq.defi-morp-init}
    \phi(i_\Ac)=i_\Mc \\
    \label{eq.defi-morp-exis}
    \forall s\in{Q_\Ac},~\forall a\in\intalph\quad (s\cdot a) \text{ exists in }\Ac \iff (\phi(s)\cdot a) \text{ exists in }\Mc \\
    \label{eq.defi-morp-tran}
    \forall s,s'\in Q_\Ac,~ \forall a\in\intalph\quad
    s\rpath{a}s' \text{ in }\Ac \implies \phi(s)\rpath{a} \phi(s') \text{ in }\Mc \\
    \label{eq.defi-morp-fina}
    F_\Ac=\phi^{-1}(F_\Mc)
  \end{gather}
\end{definition}

\begin{definition}
If a function~$\phi$ satisfies~$(\ref{eq.defi-morp-init})$, $(\ref{eq.defi-morp-exis})$
  and $(\ref{eq.defi-morp-tran})$
  but not necessarily $(\ref{eq.defi-morp-fina})$, then we say that we have an \emph{automaton pseudo-morphism}.
\end{definition}
\begin{definition}
  Two states~$s,s'$ of an automaton~$\Ac$ are \emph{Nerode-equivalent} if, for every
  word~$u$,~$(s\cdot u)$ exists and is final if and only if~$(s'\cdot u)$ exists and is final.
\end{definition}
The next result is classical. See, for instance, \cite{Sakarovitch2009}.
\begin{theorem}[Myhill--Nerode]
  Let~$\Ac$ be a complete automaton.
  Among all the complete automata accepting~$L(\Ac)$, up to isomorphism, there exists a unique one with a minimal
  number of states, called the \emph{minimisation of~$\Ac$}. Moreover, if~$\Mc$ denotes the minimisation of~$\Ac$, then there exists an automaton
  morphism~$\phi:\Ac\rightarrow\Mc$ (called the \emph{minimisation morphism}) such that
  \begin{equation}
    \forall s,s'\in\Ac \quad \phi(s)=\phi(s') \iff \text{$s$ and $s'$ are Nerode-equivalent}.
  \end{equation}
\end{theorem}
%

%
%
If $\Ac$ is an automaton and $u$ is a word, we write $(\Ac\cdot u)$ as a shorthand for $(i_\Ac\cdot u)$, i.e., the state reached by the run of~$u$ in~$\Ac$.
\begin{lemma}\label{l.cara-pseu-morp}
  Let~$\Ac$ and~$\Mc$ be two complete (and accessible) automata.
  There exists a pseudo-morphism~${\Ac\rightarrow \Mc}$
  if and only if every pair of words~$u$, $u'$ such that~$(\Mc\cdot u)\neq(\Mc\cdot u')$
  also satisfies~$(\Ac\cdot u)\neq(\Ac\cdot u')$.
\end{lemma}
\begin{proof}
  Forward direction.
  Since a pseudo-morphism $\phi$ respects transitions and the initial state, it follows
  that, for every word~$u$,~$(\Mc \cdot u) = \phi(\Ac\cdot u)$.
  The statement follows immediately.

  \medskip

  Backward direction.
  For every state~$s$, we choose a word~$u_s$ such that~$(\Ac \cdot u_s) = s$ (such a word exists because $\Ac$ is accessible).
  We define a function~$\phi: Q_\Ac \rightarrow Q_\Mc$ as follows.
  For every state~$s\in Q_\Ac$,~$\phi(s)=(\Mc\cdot u_s)$.
  Let us show that~$\phi$ is an automaton pseudo-morphism.

  Let~$s$ be a state of~$\Ac$ and let~$u$ be a word such that~$(\Ac\cdot u)=s$.
  Since~$(\Ac\cdot u) = (\Ac\cdot u_s)$, the hypothesis implies~$(\Mc\cdot u) = (\Mc \cdot
  u_s)$. The definition of~$\phi$ is therefore independent of the choice of the words~$u_s$.

  In particular,~$\phi(i_\Ac)=(\Mc\cdot u_{i_\Ac}) = (\Mc \cdot \varepsilon) = i_\Mc$
  hence~$\phi$ satisfies (\ref{eq.defi-morp-init}).
  Moreover, since both~$\Ac$ and~$\Mc$ are complete, and since~$\phi$ is a
  total function,~$\phi$ also satisfies (\ref{eq.defi-morp-exis}).
  Let~$t\rpath{a} t'$ be a transition of~$\Ac$.
  By definition~$\phi(t)=(\Mc \cdot u_t)$ and since the definition
  of~$\phi$ does not depend on the choice of the words $u_s$, we may assume that~$u_{t'}=u_{t}a$.
  It then follows that
  \begin{equation*}
      \phi(t')=(\Mc \cdot (u_ta)) = ((\Mc\cdot u_t)\cdot a)= \phi(t)\cdot a~.
  \end{equation*}
  In other words, $\phi(t)\rpath{a}  \phi(t')$ is a transition of $\Mc$.
\end{proof}
%

\subsection{Ultimately-equivalent states}

Our decision procedure involves the determination of
ultimately-equivalent states defined as follows.

\begin{definition}
  Let~$\Ac$ be an automaton over~$\intalph$. Let $m\ge 1$ be an integer.
Two states~$s,s'$ of $\Ac$ are \emph{$m$-ultimately-equivalent} if
\begin{equation*}
    \forall u\in\intalph^* \quad \wlen{u}\geq m \implies (s\cdot u) = (s' \cdot u) ~.
\end{equation*}
Two states are \emph{ultimately-equivalent} if they are $m$-ultimately-equivalent for some $m\ge 1$.
\end{definition}
\begin{numremark}
Note that ultimate-equivalence is indeed an equivalence relation:
if~$s$ and~$s'$ are $m$-ultimately-equivalent while~$s'$
and~$s''$ are $m'$-ultimately-equivalent, then~$s$ and~$s''$ are $\max(m,m')$-ultimately-equivalent.
\end{numremark}
Considering an automaton~$\Ac$ over~$\intalph$, the computation of this relation is easy.
Let us build a directed graph~$\Gc=(V,E)$ as follows.
  The vertex-set is~$V = {Q_\Ac}\times {Q_\Ac}$ and the edge set is:
  \begin{multline}\label{eq.defi-Gc}
    \forall (s,t),(s',t') \in V,~s\neq t \\ (s,t) \rightarrow (s',t')  \text{ in } \Gc \quad \iff\exists a\in\intalph \text{ such that
    $\Ac$ features }\left\{\begin{array}{l}
        s\rpath{a}s'\\
        t\rpath{a}t'
                                                            \end{array}\right.~.
  \end{multline}
  In particular, vertices of the form~$(s,s)$
  never qualify for the above condition and thus never have outgoing edges. Observe that two distinct states~$s,t$ of~$\Ac$ are ultimately-equivalent if and
  only if~$(s,t)$ may not reach in~$\Gc$ a strongly connected component.

Computing the strongly connected components of a graph is done in linear time (see, for instance, Tarjan's algorithm \cite{Tarjan1972}). Hence, the set of the pairs of states of~$\Ac$ that are ultimately-equivalent may be decided in time~$O(bn^2)$. This complexity can be improved as follows.

\begin{proposition}[B\'eal, Crochemore, \cite{Beal&Crochemore2007}]\label{p.ue-com}
   Let~$\Ac$ be an automaton over~$\intalph$.
   We write~$n$ the number of states of~$\Ac$.
   The set of the pairs of states of~$\Ac$ that are ultimately-equivalent may be decided in time~$O(bn \log n)$.
 \end{proposition}

 \begin{proof}[Sketch]
We take verbatim the algorithm in \cite{Beal&Crochemore2007}. Start from the trivial partition and iteratively merge states. Each step of the algorithm consists in merging two states that are $1$-ultimately-equivalent. The purpose of B\'eal and Crochemore was to show that starting with a so-called AFT automaton $\Ac$, the result is the minimisation of $\Ac$. Starting with any automaton $\Ac$, the resulting automaton is not necessarily minimal. However, one can observe that its states are precisely the ultimate-equivalence classes of $\Ac$.
 \end{proof}

As a direct consequence of the definition of an automaton morphism,  ultimate-equivalence commutes with automaton morphisms.

\begin{lemma}\label{l.ulti-equi-comm-morp}
  Let~$\Ac$ and~$\Mc$ be two automata such that there is an automaton
  morphism~$\phi:\Ac\rightarrow \Mc$.
  Let~$s$ and~$s'$ be two states of~$\Ac$ that are ultimately-equivalent (w.r.t.\@~$\Ac$),
  then~$\phi(s)$ and~$\phi(s')$ are also ultimately-equivalent (w.r.t.\@~$\Mc$).
\end{lemma}

\section{Purely periodic $b$-recognisable sets}

\label{s.pp}

\begin{notation}\label{n.nota}
Let $p>0$ and $b>1$ be two integers.
Throughout this section, the quantities $k,d,j,\psi$ are fixed as follows.
\begin{itemize}
 \item Let $k,d$ be the unique integers such that~$p=k\,d$ where~$k$ is the greatest divisor of $p$ coprime with~$b$. In particular, the prime factors occurring in the prime decomposition of $d$ all appear in the prime decomposition of $b$. Moreover, $(k,d)=1$.
 \item Let~$j$ be the least integer such that $d$ is a divisor of~$b^j$.
    \item Since $(k,b)=1$, the order of~$b$ in~$\ZZ[k]$ is well defined and denoted by $\psi$, i.e., $b^\psi \equiv 1~[k]$.
\end{itemize}
\end{notation}

 Let $s<k$ and $t<d$ be two integers. We let $\crl{s,t}$ denote the (unique) integer of~$\ZZ[p]$ congruent to~$s$ modulo~$k$ and~$t$ modulo~$d$. Note that if $n$ is an integer less than $p$, then $n=\crl{n\mod{k},n\mod{d}}$ where $n\mod{k}$ denote the remainder of the division of $n$ by $k$.

\subsection{The automaton~$\ARpb$ and its minimisation}

\begin{definition}\label{d.purely-periodic}
    A subset $P$ of integers is \emph{purely periodic}, if there exist
    $p\ge 1$ and a subset $R\subseteq\{0,\ldots,p-1\}$ such that
    $P=R+p\N$.
\end{definition}
For instance, $\{0,1\}+4\N$ is purely periodic but $\{4,5\}+4\N$ is not.
Let $p\ge 1$ be an integer and $R$ be a subset of $\{0,\ldots,p-1\}$.
  We say that the parameter~$(p,R)$ is \emph{proper}, if~$p$
  is the smallest period of the purely periodic set~$R+p\N$. For instance, $(4,\{0,1\})$ is proper but $(4,\{0,2\})$ is not because $\{0,2\}+4\N=\{0\}+2\N$.

The following definition is ubiquitous when dealing with periodic sets of integers. It is an easy exercise to show that this automaton accepts base-$b$ representations of integers whose remainder modulo $p$ belongs to $R$.
\begin{definition}
  We let~$\ARpb$ denote the automaton $\ARpb = \langle \intalph, \ZZ, \delta, 0, R  \rangle$ where~$\delta$ is defined as
  \begin{equation*}
      \forall n\in\ZZ,~\forall a\in\intalph \quad n \rpath{a}nb+a~.
  \end{equation*}
\end{definition}
When we are only interested in the transitions of the automaton $\ARpb[R][p]$, it is sometimes convenient to leave the set of final states unspecified. In that case, we write $\ARpb[?][p]$ for the automaton where the final/non-final status of the states is not set.
\begin{example}\label{exa:12-57}
Figure~\ref{f.A-57-12-main} shows~$\ARpb[\{5,\,7\}][12]$ in base $2$. Transitions with label $1$ (resp. $0$) are represented with bold (resp. thin) edges.
\begin{figure}[ht]
  \hfill%
  \hspace*{16.8mm}
  \hspace*{16mm}%
  \begin{subfigure}{0.4\textwidth}
    \includegraphics[scale=\AutomatonScale]{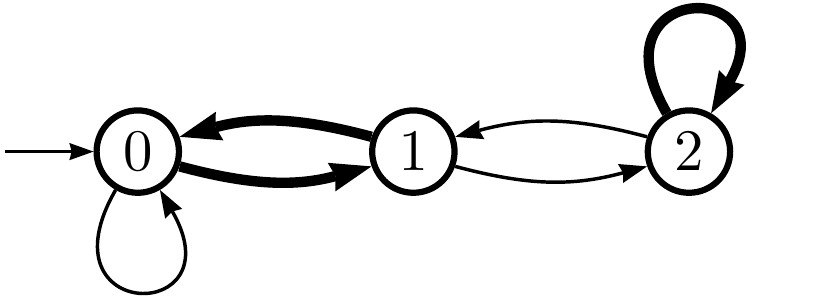}
    \caption{$\ARpb[?][3]$}
  \end{subfigure}
  \hfill\hspace*{0cm}%

  \hfill
  \begin{subfigure}{0.4\textwidth}
    \includegraphics[scale=\AutomatonScale]{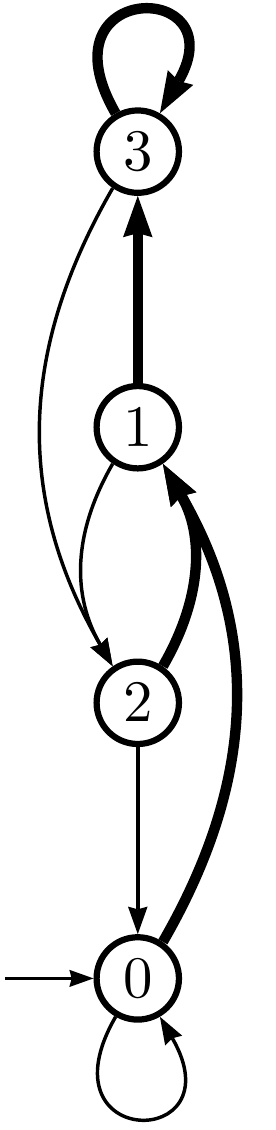}
    \caption{$\ARpb[?][4]$}
  \end{subfigure}\hspace*{0mm plus -1fill}
  \begin{subfigure}{0.4\textwidth}
    \includegraphics[scale=\AutomatonScale]{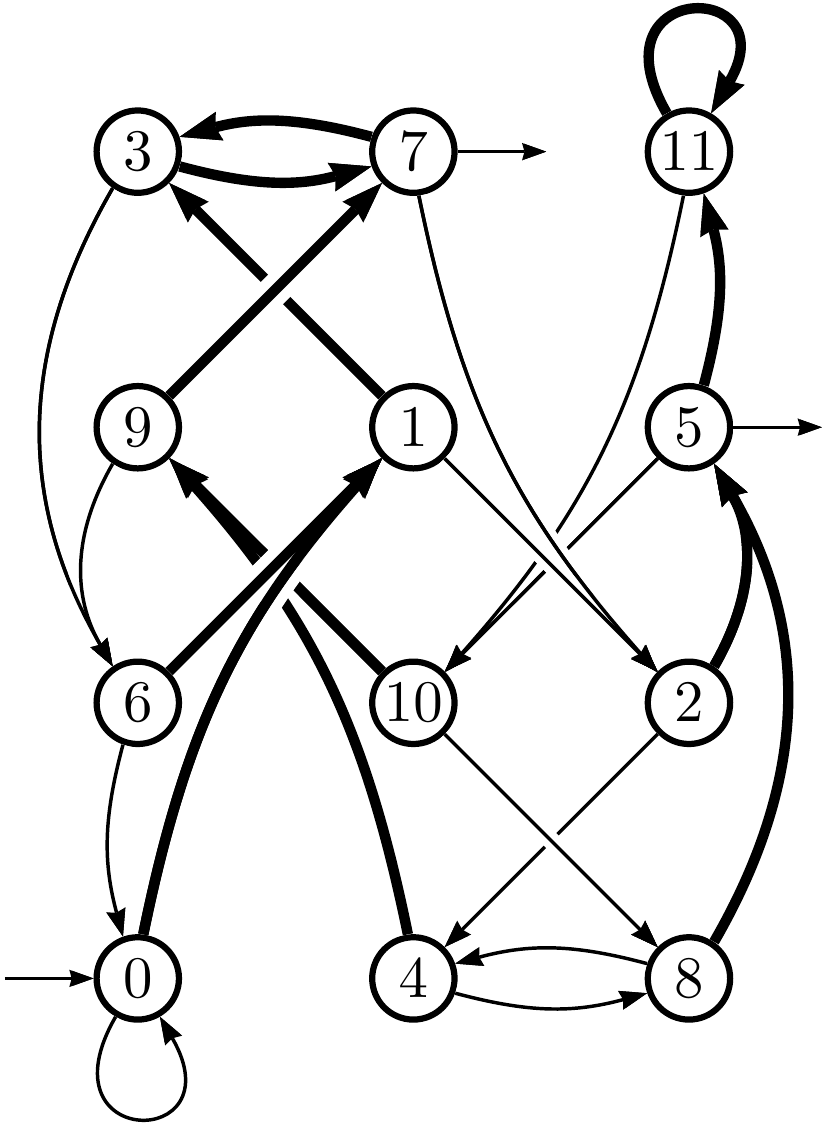}
    \caption{${{\ARpb[\{5,7\}][12]}}$}
    \label{f.A-57-12-main}
  \end{subfigure}
  \hfill\hspace*{0cm}

\caption{The automaton~${{\ARpb[\{5,7\}][12]}}$, as the product automaton of~${{\ARpb[?][4]}}$ by~${{\ARpb[?][3]}}$}
\label{f.A-57-12}
\end{figure}
\end{example}
As can be seen, for instance, in Figure~\ref{fig:minimisation57-12}, the automaton $\ARpb$ is not necessarily minimal.
\begin{lemma}\label{l.ARpb-corr}
  For every word~$u\in\intalph^*$,~$(\ARpb\cdot u)=(\val{u}\mod{p})=\crl{\val{u}\mod{k},\val{u}\mod{d}}$.
\end{lemma}
\begin{proof} This follows directly from the definition of the transition function of $\ARpb[?]$.
\end{proof}
\begin{property}\label{pp.sc}
The automaton $\ARpb$ is strongly connected.
\end{property}
\begin{proof}
Let $n,m$ be two states. The state $n$ is of the form $\crl{i,i'}$. Let $u$ be a word satisfying
\begin{equation*}
\val{u}\equiv\crl{k-i,0}[p]~,\quad \wlen{u}\ge j\quad \text{and}\quad \wlen{u}\equiv 0[\psi]~.
\end{equation*}
The last two conditions are easily satisfied by adding a suitable number of leading zeroes.
Reading $u$ from $n$ leads to the initial state $0$. Obviously, reading $\rep{m}$ from $0$ leads to~$m$.
\end{proof}
The next lemma states that the automaton~$\ARpb[?]$
is the product automaton~$\ARpb[?][k] \times \ARpb[?][d]$. This easily follows from the Chinese remainder theorem and
Lemma~\ref{l.ARpb-corr}.
\begin{lemma}\label{l.ARpb-x}
  For all integers~$s,s'\in\ZZ[k]$,~$t,t'\in\ZZ[d]$ and every word~$u\in\intalph^*$,
  \begin{equation*}
 \crl{s,t} \rpath{u} \crl{s',t'} \text{ in }\ARpb[?] \iff \left\{ \begin{array}{l}
                                              s \rpath{u} s' \text{ in }\ARpb[?][k] \\
                                              t \rpath{u} t' \text{ in }\ARpb[?][d]
                                            \end{array}\right.
\end{equation*}
\end{lemma}
The fact that $k$ is coprime with $b$ implies the following result.
\begin{lemma}\label{l.A?kb-grou}
  With the definition introduced in Notation~\ref{n.nota}, the automaton~$\ARpb[?][k]$ is a group automaton:  each letter induces a
  permutation on the set of states.
\end{lemma}
\begin{proof}
  Since~$k$ is coprime with~$b$, the function~$f_0:\ZZ[k]\rightarrow\ZZ[k]$
  defined by~$s\mapsto sb$ is a permutation of~$\ZZ[k]$.
  Hence, so is the function~$f_a$ defined by~$s\mapsto (sb+a)$, for
  every letter~$a\in\intalph$.
  The action of~$a$ in~$\ARpb[?][k]$ is exactly~$f_a$, a
  permutation of the states.
\end{proof}
%

\subsection{Nerode-equivalence and ultimate-equivalence in~$\ARpb$}

Within the setting of Example~\ref{exa:12-57} where rows (resp. columns) of the product automaton $\ARpb[R][p]\approx \ARpb[?][d]\times \ARpb[?][k]$ correspond to the equivalence classes modulo $d$ (resp. modulo $k$), the forthcoming Proposition~\ref{p.not-nero-equi} shows that Nerode-equivalent states in $\ARpb[R][p]$ must belong to the same column. See, for instance, Figure~\ref{fig:minimisation57-12}.
Then, we show that all states belonging to the same column are ultimately-equivalent.

\begin{lemma}\label{l.not-nero-equi}
  If~$(p,R)$ is proper, then for all distinct integers~$i$ and~$i'$,~$0\leq i,i'<k$, the
  states~$i d$ and~$i'd$
  are not Nerode-equivalent.
\end{lemma}
\begin{proof}
  Since~$(p,R)$ is proper and~$id\neq i'd$, there exists an integer~$m$ such that~$(i d+m)\in R+p\N$
  and~$(i' d+m)\notin R+p\N$.
  We let~$u$ denote a word such that~$\val{u} = m$ and~$\wlen{u} \equiv 0 [\psi]$
  (in other words,~$u$ is the word~$\rep{m}$ padded with an appropriate number of 0's);
  it thus holds that~$b^{\wlen{u}}\equiv 1~[k]$.
  Reading the word~$u$ respectively from the states~$id$ and~$i'd$ leads to the states:
  \begin{equation*}
      id\cdot u = i d b^{\wlen{u}}+m\quad\text{and}\quad i'd\cdot u = i' d b^{\wlen{u}}+m~.
  \end{equation*}
  The integer~$(i d b^{\wlen{u}}+m)$ is congruent to~$(i d+m)$
  modulo~$k$ (since~$b^{\wlen{u}}\equiv 1~[k]$) as well as modulo~$d$ (since both are obviously
  congruent to~$m$) hence modulo~$p$.
  The same reasoning also applies to the second state, finally yielding:
  \begin{equation*}
      id\cdot u = i d+m\quad\text{and}\quad i'd\cdot u = i' d+m~.
  \end{equation*}
  The first state belongs to~$R$ and is thus final while the second does not belong to~$R$
  and thus is not final.
  The word~$u$ is then a witness of the fact that~$id$ and~$i'd$ are not Nerode-equivalent.
\end{proof}

\begin{proposition}\label{p.not-nero-equi}
    Let~$(p,R)$ be proper. If~$i$ and~$i'$ are Nerode-equivalent states, then they are congruent modulo~$k$.
\end{proposition}
\begin{proof}
  Proof by contrapositive.
  Let~$i$ and~$i'$ be two states that are not congruent modulo~$k$.
  By definition of $j$, see Notation~\ref{n.nota}, the states~$(i\cdot 0^j)$ and~$(i'\cdot 0^j)$ are both congruent to $0$
  modulo~$d$.
  However the operation~$i\mapsto ib$ is a permutation of~$\ZZ[k]$,
  hence~$(i\cdot 0^j)$ and~$(i'\cdot 0^j)$ are not congruent modulo~$k$.
  It follows that~$(i\cdot 0^j)=ld$ and~$(i'\cdot 0^j)=l'd$ for some
  distinct~$l,l'\in\ZZ[k]$.
  Lemma \ref{l.not-nero-equi} then yields that these
  states are not Nerode-equivalent, hence that~$i$ and~$i'$ are not either.
\end{proof}

\begin{lemma}\label{l.equiv-k-ulti-equiv}
  Let~$s$ and~$s'$ be two states of~$\ARpb$. With the definition introduced in Notation~\ref{n.nota}, if~$s\equiv s'[k]$,
  then~$s$ and~$s'$ are $j$-ultimately-equivalent.
\end{lemma}
\begin{proof}
  Let~$u$ be any word of length~$j$.
  Since~$s$ and~$s'$ are congruent modulo~$k$, there exists~$i\in\ZZ[k]$
  and~$l,l'\in\ZZ[d]$ such that~$s=\crl{i,l}$ and~$s'=\crl{i,l'}$.
  Then, from Lemma~\ref{l.ARpb-x} and using the fact that $l b^j\equiv 0~[d]$, we get
  \begin{equation*}
 (s\cdot u) = \crl{ i b^j + \val{u},\, l b^j +\val{u}}
                = \crl{ i b^j + \val{u},\, \val{u}}~.
  \end{equation*}
  Similarly~$(s'\cdot u) = \crl{ i b^j + \val{u},\, \val{u}} = (s\cdot u)$.
\end{proof}

\begin{figure}[ht]

  \hspace*{0cm plus 1fill}
  \begin{subfigure}{0.4\linewidth}
    \includegraphics[scale=\AutomatonScale]{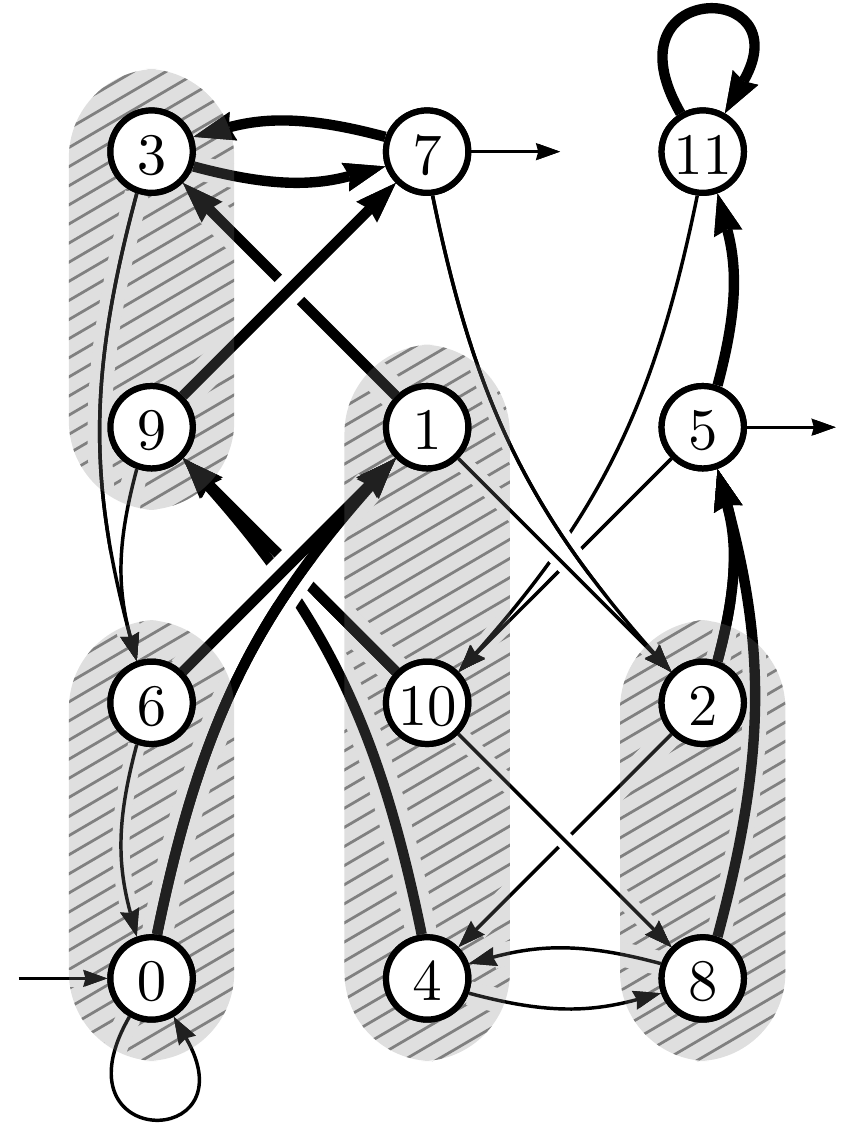}
    \caption{Nerode-equivalence classes of~{{{$\ARpb[\{5,\,7\}][12]$}}}}
  \end{subfigure}
  \hspace*{0cm plus 2fill}%
  \begin{subfigure}{0.4\linewidth}
    \includegraphics[scale=\AutomatonScale]{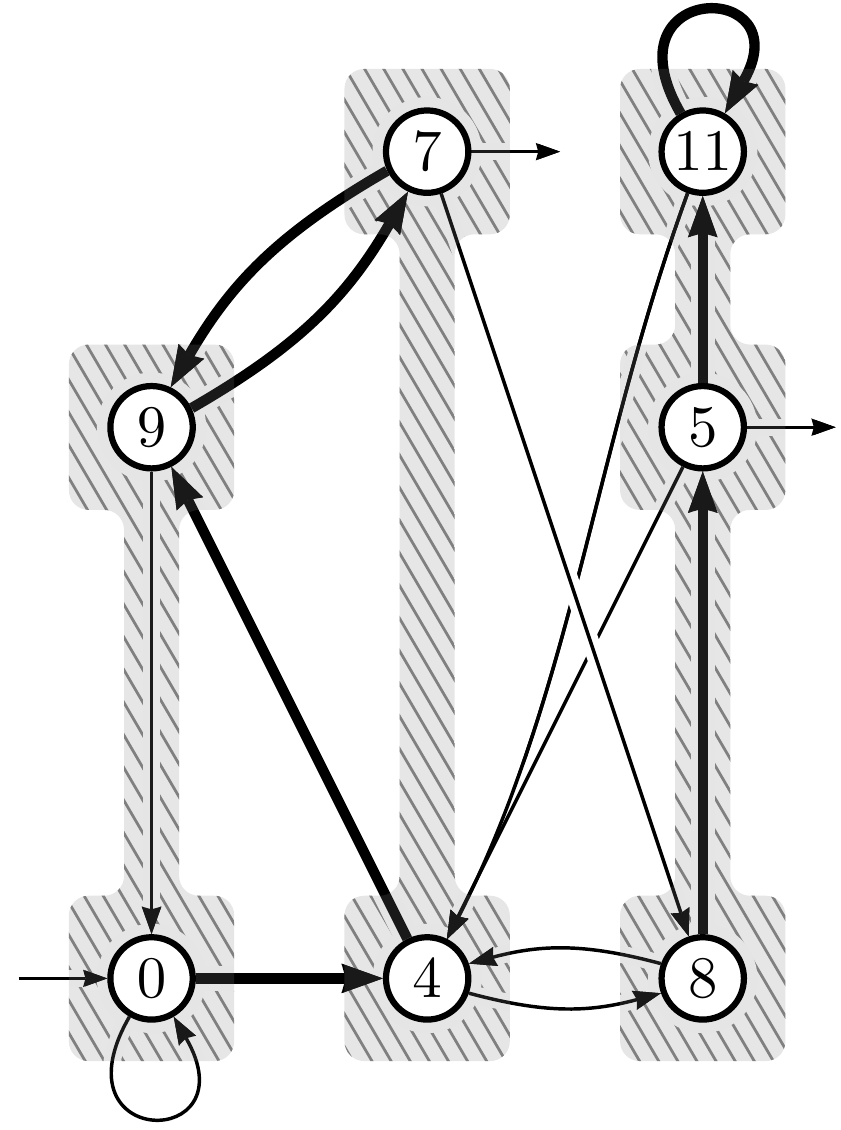}%
    \caption{Pseudo-morphism equivalence classes in the minimisation of~{{{$\ARpb[\{5,\,7\}][12]$}}}}
  \end{subfigure}
  \hspace*{0cm plus 1fill}

  \caption{Minimisation morphism of~$\ARpb[\{5,\,7\}][12]$ and pseudo-morphism of its minimisation}\label{fig:minimisation57-12}
\end{figure}

\subsection{Circuits labelled by the digit~$0$}

  A circuit whose every arc is labelled by the digit~$0$ is called for short a~\emph{0-circuit}. For instance, the automaton $\ARpb[\{5,7\}][12]$ depicted in Figure~\ref{f.A-57-12} has two such circuits: one reduced to the state $0$ and one made of the states $4$ and $8$. We will see that the number of states belonging to $0$-circuits has a special meaning.
  \begin{lemma}\label{l.0-circ}
    A state of~$\ARpb$ is a multiple of~$d$ if and only if it belongs
    to a~$0$-circuit.
  \end{lemma}
  \begin{proof}
    Forward direction.
    It is enough to show that every state of the form~$i d$, for~$i\in\ZZ[k]$, has a predecessor
    by~$0$ of the form~$i'd$,~$i'\in\ZZ[k]$.
    Simple arithmetic yields that~$(b^{-1}i) d $ is suitable, where~$b^{-1}$ is the inverse
    of~$b$ in~$\ZZ[k]$.
    Backward direction. Proof by contrapositive.
    Let~$s$ be a state which is not a multiple of~$d$.
    The state~$(s\cdot 0^j)$ is a multiple of~$d$. Therefore, for every integer~$i\geq j$, the
    state~$(s\cdot 0^{i})$
    is a multiple of~$d$, hence is not equal to~$s$.
    Since~$\ARpb$ is deterministic,~$(s\cdot 0^i)$ cannot be equal to~$s$ for~$i<j$ either.
  \end{proof}
  The next proposition follows from Lemmas \ref{l.0-circ} and \ref{l.not-nero-equi}. Recall that $k$ is the largest integer coprime with $b$ such that $p=k\, d$ and $d\ge 1$ (see Notation~\ref{n.nota}).
  \begin{proposition}\label{p.mini-k-stat}
      If~$(p,R)$ is proper, the minimisation of~$\ARpb$ possesses
      exactly~$k$ states that belong to $0$-circuits.
  \end{proposition}
  %

\section{Characterisation of  automata accepting purely periodic sets}

The next result will allow us to decide whether a deterministic automaton $\Ac$ over $\intalph$, given as input, is such that $\val{L(\Ac)}$ is a purely periodic set of integers, i.e., whether or not it is of the form $R+p\N$ for some $R$ and $p$.

\begin{theorem}\label{t.pp}
  Let~$b>1$ be a base and~$\Ac$ be a minimal automaton over~$\intalph$ bearing a self-loop labelled by~$0$ on the initial state.
  Let~$\ell$ be the number of states in~$\Ac$ that belong to $0$-circuits.
  The automaton~$\Ac$ accepts by value a purely periodic set of integers if and only if
  the following two conditions are fulfilled.
  \begin{enumerate}
    \renewcommand{\theenumi}{\alph{enumi}}
    \item\sublabel{t.pp-pseu-morp} There exists a pseudo-morphism~$\phi:\Ac\rightarrow \ARpb[?][\ell]$.
    \item\sublabel{t.pp-ulti-equi} The equivalence relation induced by~$\phi$ is a refinement
    of the ultimate-equivalence relation.
  \end{enumerate}
\end{theorem}

\begin{proof}[Proof of forward direction]
    Since~$\Ac$ accepts by value a purely periodic set of integers,
    there exists a smallest period~$p$ and a
    remainder-set~$R\subseteq\{0,\ldots,p-1\}$ such that~$L(\Ac)=0^*\rep{R+p\N}$.
    Note that $(p,R)$ is proper by choice of~$p$. We make use of
    Notation~\ref{n.nota}. In particular, $k$ is the greatest divisor
    of~$p$ that is coprime with~$b$.
  Since~$\Ac$ is minimal, it is isomorphic to the minimisation of any automaton accepting~$L(\Ac)$, in particular, to the minimisation of~$\ARpb$.
  It then follows from Proposition~\ref{p.mini-k-stat} that~$\ell=k$.

  \medskip

To prove that there exists a pseudo-morphism~$\phi:\Ac\rightarrow \ARpb[?][k]$,
we will apply Lemma~\ref{l.cara-pseu-morp}. Let~$u$,$u'$ be two words such
that~$(\ARpb[?][k]\cdot u)\neq(\ARpb[?][k]\cdot u')$. Let us show that
$(\Ac\cdot u) \neq (\Ac\cdot u')$.
Since~$(\ARpb[?][k]\cdot u)\neq(\ARpb[?][k]\cdot u')$, we have that~$\val{u}\not\equiv \val{u'}~[k]$. Due to Lemma~\ref{l.ARpb-corr},~$(\ARpb\cdot u)$ and~$(\ARpb\cdot u')$ are not congruent
  modulo~$k$.
  It then follows from Proposition~\ref{p.not-nero-equi} that the states~${(\ARpb\cdot u)}$
  and~${(\ARpb\cdot u')}$ are not Nerode-equivalent, which implies
  that~$(\Ac\cdot u) \neq (\Ac\cdot u')$ because~$\Ac$ is the minimisation
  of~$\ARpb$.

  \medskip

  Let~$s$ and~$s'$ be two states of~$\Ac$ such that~$\phi(s)=\phi(s')$. We have to show that~$s$ and $s'$ are ultimately-equivalent.
  Let~$u$ and~$u'$ be two words such $(\Ac\cdot u)=s$ and $(\Ac\cdot u')=s'$. Since $\phi$ is a pseudo-morphism, we get that
  \begin{equation*}
      (\ARpb[?][k]\cdot u)=\phi(s)=\phi(s')= (\ARpb[?][k]\cdot u')
  \end{equation*}
  and so~$\val{u}\equiv\val{u'}~[k]$.
  Applying Lemma~\ref{l.ARpb-corr} yields that the states~$(\ARpb\cdot u)$
  and~$(\ARpb\cdot u')$ are congruent modulo~$k$,
  and by Lemma~\ref{l.equiv-k-ulti-equiv}, these states
  are ultimately-equivalent.
Since $\Ac$ is the minimisation of~$\ARpb$, we have an automaton morphism $\ARpb\to\Ac$.
  Finally, since ultimate-equivalence commutes with automaton morphism
  (Lemma \ref{l.ulti-equi-comm-morp}),~$(\Ac\cdot u)=s$ and~$(\Ac\cdot u')=s'$ are
  ultimately-equivalent.
\end{proof}

\begin{proof}[Proof of backward direction]
  By assumption, for all $i\in\ZZ[\ell]$, the states in~$\phi^{-1}(i)$ are ultimately-equivalent.
  For every integer~$i\in\ZZ[\ell]$, we let~$m_i$ denote the least integer such that, for all $s,s'$ in $\phi^{-1}(i)$, $(s\cdot u)=(s'\cdot u)$ whenever $|u|\ge m_i$. Let $m=\max\{m_i\mid{i\in\ZZ[\ell]}\}$.

  \medskip

  Let~$u,u'$ be two words with respective values that are congruent modulo~$\ell b^m$. Note that, in particular, $\val{u}$ and $\val{u'}$ are thus congruent modulo~$b^m$.
  Let us show that~$u$ and~$u'$ reach the same state in~$\Ac$.
  Since~$\Ac$ bears a self-loop labelled by~$0$ on the initial state,
  the word~$0^mu$ is such that~$\val{0^mu}=\val{u}$ and~$\Ac\cdot 0^mu =\Ac\cdot u$.
  We may thus assume that~$u$ and~$u'$ are longer than~$m$.
  There exist factorisations~$u=vw$ and~$u'=v'w'$ such that the
  lengths of~$w$ and~$w'$ are both equal to~$m$.
  Since~$\val{u}$ and~$\val{u'}$ are congruent modulo~$b^m$,~$w$ and~$w'$ are equal: $u=vw$,~$u'=v'w$.

Assume without loss of generality that $\val{u}\ge \val{u'}$. Hence $\val{u}-\val{u'}=(\val{v}-\val{v'})b^m$ is congruent to $0$ modulo $\ell b^m$. We deduce that $\val{v}$ and $\val{v'}$ are congruent modulo $\ell$.
By Lemma~\ref{l.ARpb-corr}, the respective runs of $v$ and $v'$ in~$\ARpb[?][\ell]$
  reach the same state:~$(\ARpb[?][\ell]\cdot v) = (\ARpb[?][\ell]\cdot v')$.
From assumption \ref{t.pp-pseu-morp*}, we get $\phi(\Ac\cdot v)=\phi(\Ac\cdot v')$.
  %
  %
  %
  In other words, the states~$(\Ac\cdot v)$ and~$(\Ac\cdot v')$ are~$\phi$-equivalent. Hence,
  by assumption~\ref{t.pp-ulti-equi*}, they are $m_i$-ultimately-equivalent.
  Since~$\wlen{w}=m\geq m_i$ (by choice of~$m$), we get that~$(\Ac\cdot v \cdot w)=(\Ac\cdot v'\cdot w)$:
  the run in~$\Ac$ of the words~$u=vw$ and~$u'=v'w$ indeed reach the same state.

  \medskip

  We have just shown that words whose values are congruent modulo~$\ell b^m$
  have runs in~$\Ac$ reaching the same states, hence either all are accepted by~$\Ac$
  or none of them are.
  The run of a word~$u$ is then accepted by~$\Ac$ if and only if~$\rep{\val{u}\mod(\ell b^m)}$ is.
  Finally, a word~$u$ is accepted by~$\Ac$
  if and only if~$\val{u}\mod(\ell b^m)$ belongs to the set~$R\subseteq\{0,\ldots,\ell b^m-1\}$, defined by
  \begin{equation*}
      R=\{~i\in\ZZ[\ell b^m]~|~(\Ac\cdot\rep{i})\text{ is final}~\}~.\qedhere
  \end{equation*}
\end{proof}

\begin{numremark}\label{r.l=k}
  In the proof of the forward direction, it was stated that~$\ell=k$ (where~$k$ is the greatest divisor of the period which is coprime with the base).
  It is also the case in the backward direction. Indeed, the automaton~$\Ac$ is shown to
  accept a purely periodic set of integers.
  Let~$(p,R)$ denotes the \textbf{proper} parameter of this set (it is not necessarily the one given in the proof).
  Since~$\Ac$ is minimal, it is the quotient of~$\ARpb[R][p]$.
  It then follows from Proposition~\ref{p.mini-k-stat} that,~$\ell$, the number of states belonging to
  $0$-circuits, is equal to~$k$, the greatest divisor of the period which is coprime with the base.
\end{numremark}

\subsection{Complexity and algorithmic issues}
\label{s.comp-algo}

Theorem \ref{t.pp} yields an algorithm to decide whether a given deterministic automaton~$\Ac$ accepts
by value a purely periodic set of integers:
\begin{enumerate}
  \renewcommand{\theenumi}{\arabic{enumi}}
  \addtocounter{enumi}{-1}
  \item if necessary, minimise $\Ac$ and make it complete;
  \item count the number $\ell$ of states of~$\Ac$ that belong to $0$-circuits;
  \item build the automaton~$\ARpb[?][\ell]$;
  \item construct, if it exists, the pseudo morphism~$\phi:\Ac \rightarrow\ARpb[?][\ell]$;
  \item check whether, for all~$x\in\ZZ[\ell]$, the states of~$\phi^{-1}(x)$ are ultimately-equivalent.
\end{enumerate}
Let us denote by~$n$ the number of states of~$\Ac$.
Step~$(0)$ can be carried out in~$O(bn\log n)$ time.
Steps~$(1)$,~$(2)$ can obviously be performed in~$O(bn)$ time.
A morphism between deterministic automata, if it exists, can be computed by a
single traversal of the bigger automaton; the same algorithm also works for pseudo-morphisms:
Step~$(3)$ also runs in~$O(bn)$ time.
The ultimate-equivalence classes of~$\Ac$ can be computed in time~$O(b n \log n)$ from
Proposition~\ref{p.ue-com}, hence so is the execution of Step~$(4)$.
\begin{corollary}\label{c.pp}
  Let~$b>1$ be a base and~$\Ac$ be a $n$-state deterministic automaton over~$\intalph$. It is decidable in~$O(b n \log n)$ time whether~$\Ac$ accepts by value a purely periodic set of integers.
\end{corollary}
%

%
%

%
\begin{numremark}
  Remark~\ref{r.l=k} gives a very fast rejection test.
  Indeed, before Step~(2) we may check whether the integer~$\ell$
  (computed by Step (1))  is coprime with~$b$.
  If it is not the case,~$\Ac$ may be rejected already.
\end{numremark}
%


\begin{example}
  We start with the minimal automaton~$\Ac$ depicted in Figure~\ref{fig3a}.
  Step (1) is shown in Figure~\ref{fig3b}: $\Ac$ has five states
  belonging $0$-circuits and thus, $\ell=5$.
  Step (2) then consists in constructing~$\ARpb[?][5]$, shown in~$\ref{f.A?b}$.
  There is a pseudo-morphism~$\Ac\rightarrow \ARpb[?][5]$, whose
  equivalence classes are represented in Figure~\ref{fig4}.
  Finally, one could check that Step (4) holds:
  all states belonging to the same class are $3$-ultimately-equivalent.
  Hence $\Ac$ accepts an eventually periodic set of period $2^3\times 5$.
  It is indeed the minimisation of~$\ARpb[\{0,3\}][40]$.
  \begin{figure}[ht!]
    \begin{minipage}[t]{0.5\linewidth}
      \centering
      \includegraphics[scale=\AutomatonScale]{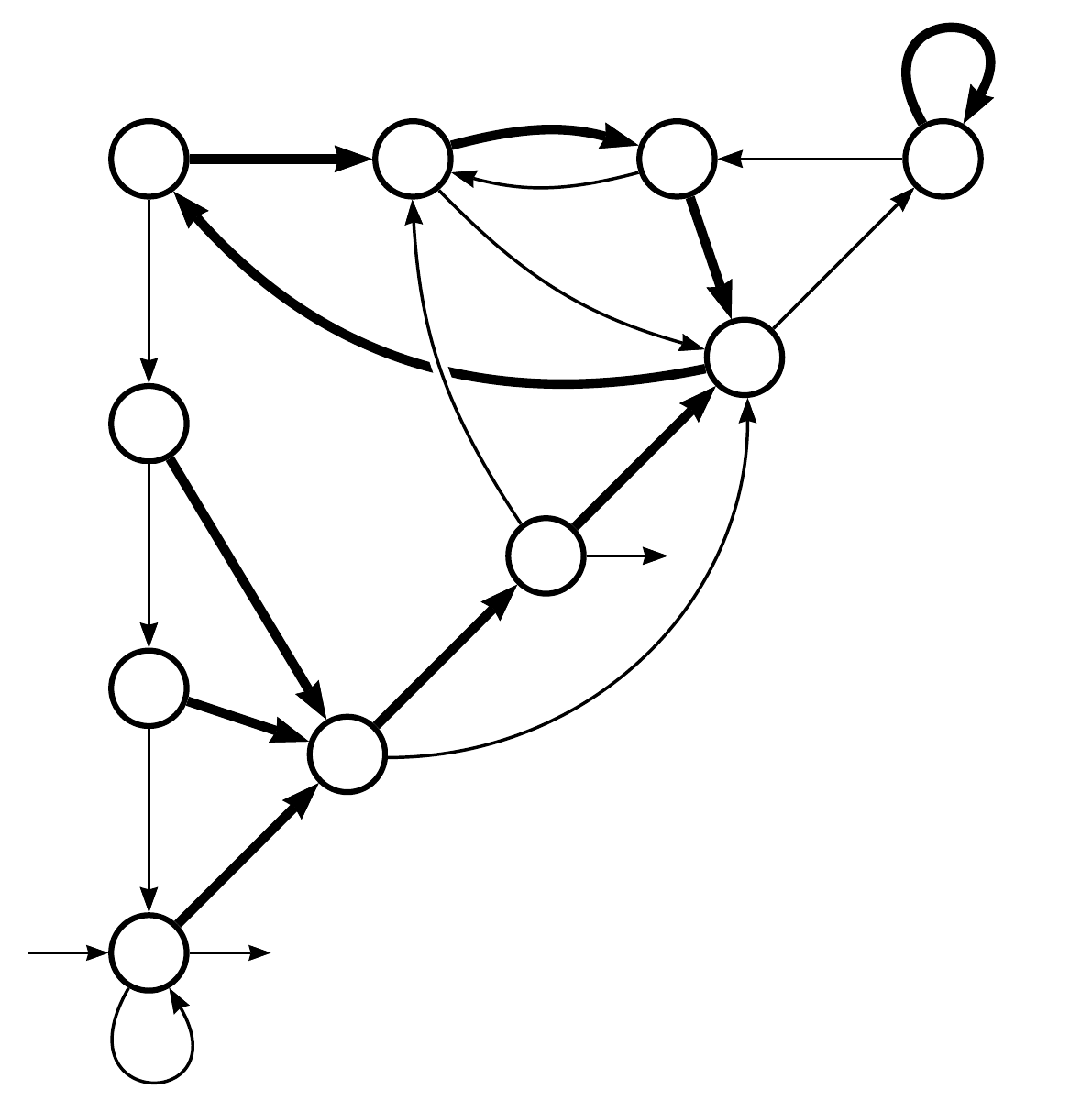}
      \captionof{figure}{An automaton~$\Ac$}
      \label{fig3a}
    \end{minipage}%
    \begin{minipage}[t]{0.5\linewidth}
      \centering
      \includegraphics[scale=\AutomatonScale]{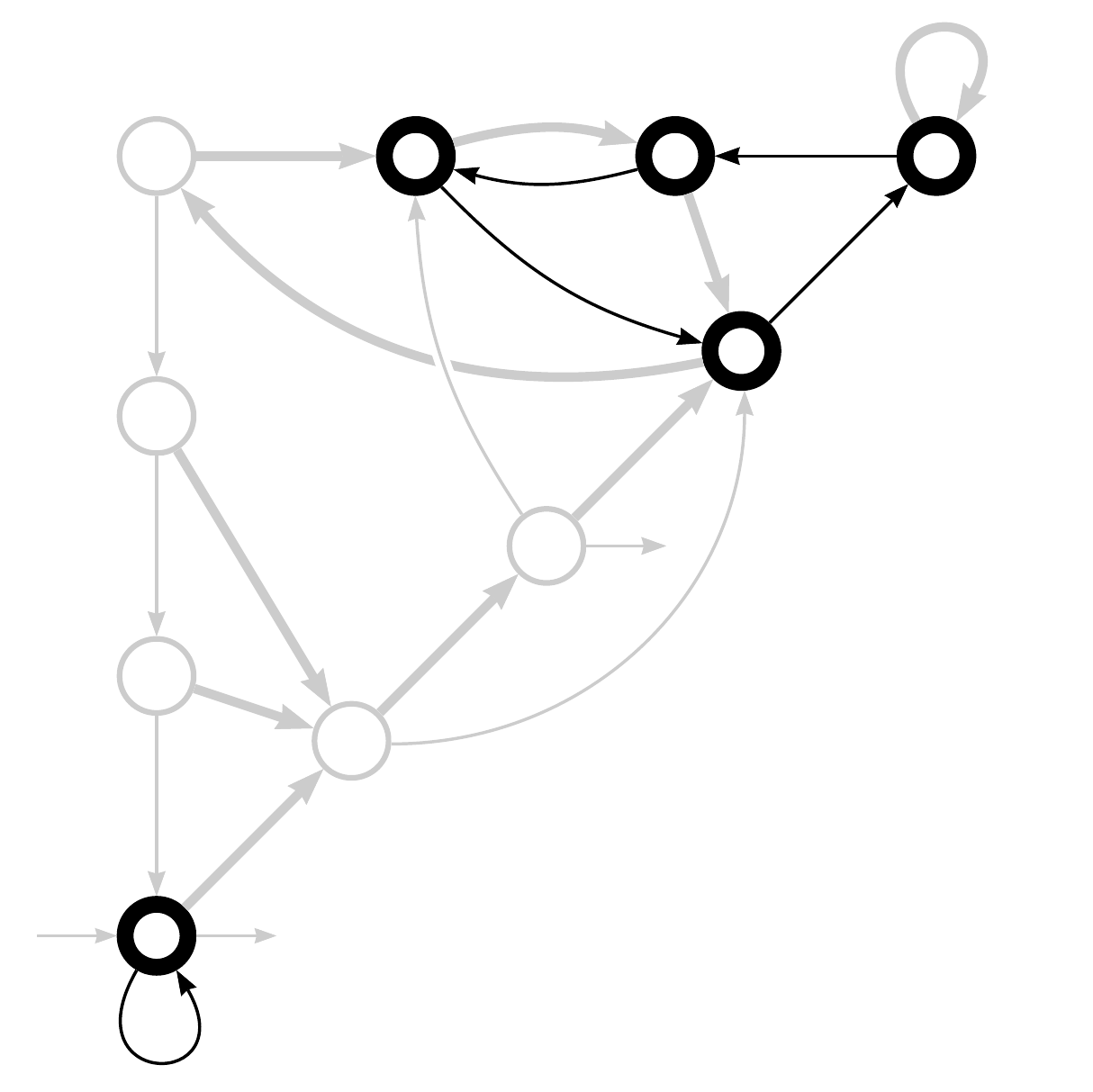}
      \captionof{figure}{The 0-circuits of~$\Ac$ have~$5$ states in total}
      \label{fig3b}
    \end{minipage}%
    \hspace*{0cm plus 1fill}
  \end{figure}
  \begin{figure}[ht!]
    \begin{minipage}[t]{0.5\linewidth}
      \centering
      \includegraphics[scale=\AutomatonScale]{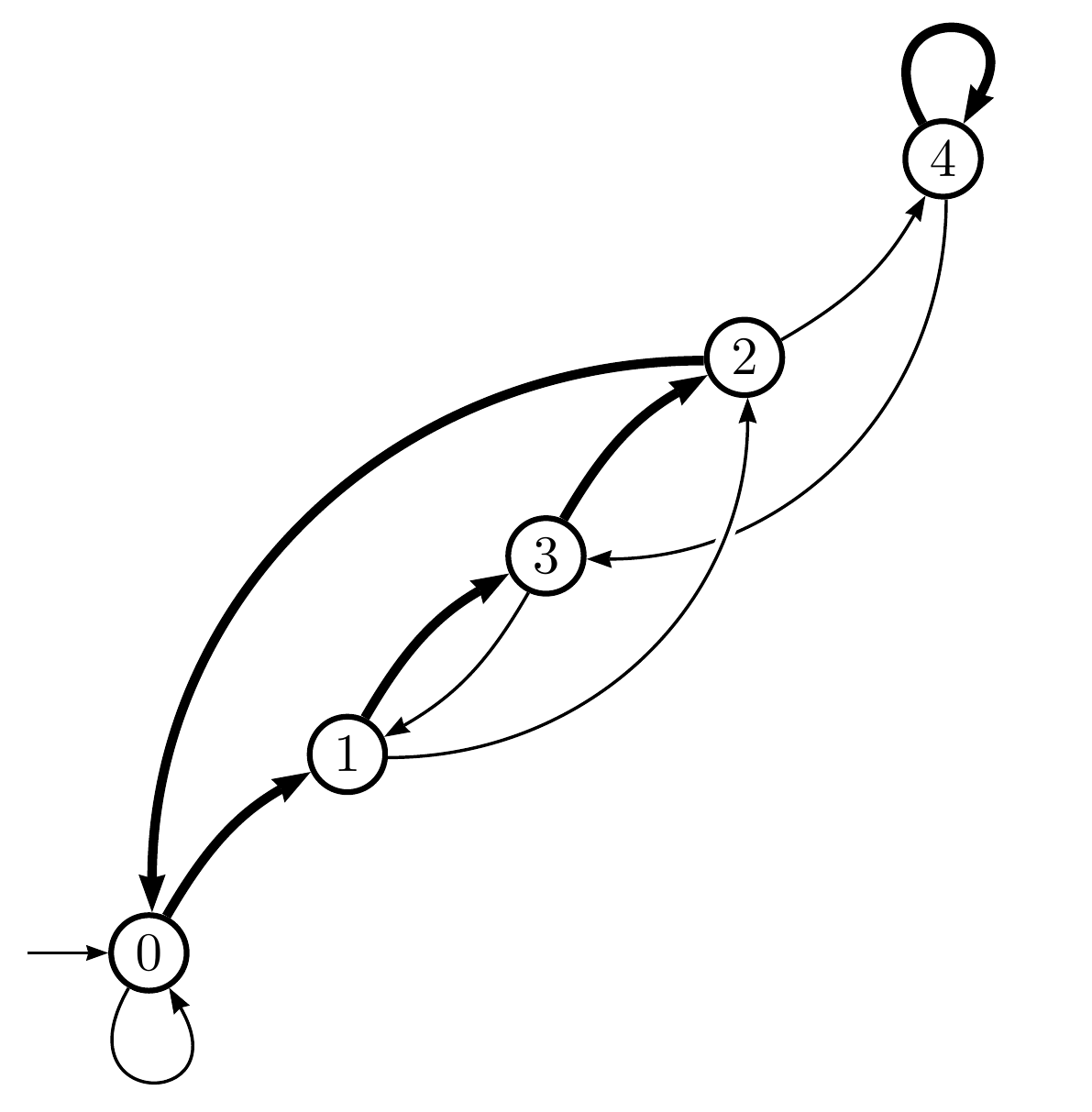}
      \captionof{figure}{The automaton~$\ARpb[?][5]$}
      \label{f.A?b}
    \end{minipage}%
    \begin{minipage}[t]{0.5\linewidth}
      \centering
      \includegraphics[scale=\AutomatonScale]{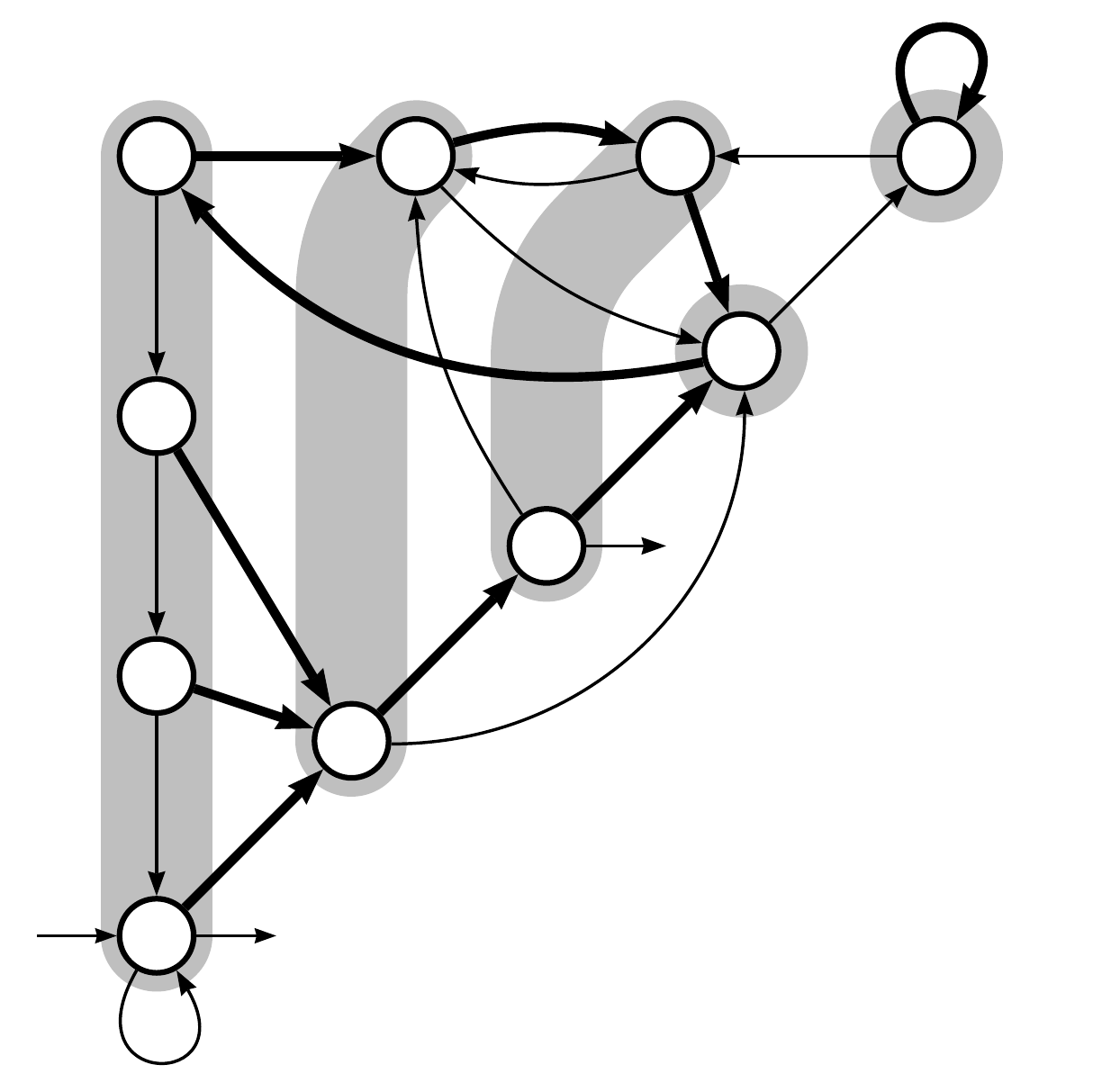}
      \captionof{figure}{Equivalence classes of the relation induced by the
            pseudo-morphism~$\Ac\rightarrow \ARpb[?][5]$}
      \label{fig4}
    \end{minipage}%
  \end{figure}
\end{example}

\section{Impurely periodic b-recognisable sets}

In this section, we will study the eventually periodic sets of integers that are not purely
periodic (see Definition~\ref{d.purely-periodic}). We say that such sets are \emph{impurely periodic}.
For denotational reasons, we will describe eventually periodic sets~$S$ with three parameters:
a period~$p$, a remainder-set~$R\subseteq \{0,\ldots,p-1\}$ and a finite set~$I\subseteq\N$ of ``mismatches'' with a purely periodic set.
Such a triplet~$(p,R,I)$ is a \emph{parameter} of~$S$ if
\begin{equation*}
    S = (R+p\N) \oplus I~,
\end{equation*}
where $\oplus$ is the \emph{exclusive disjunction} operation: an integer belongs to~$S$ if it belongs
either to~$(R+p\N)$ or to~$I$, but not both. One can also find the terminology \emph{symmetric difference} or \emph{disjunctive union} (and the notation $\Delta$).
\begin{example}
    The set $S=\{0,6\}\cup(\{4,5\}+4\N)$ can be described by the parameter $(4,\{0,1\},\{1,6\})$. Indeed, the purely periodic set $P=\{0,1\}+4\N$ and the set $S$ differ only by the fact that $1\in P\setminus S$ and $6\in S\setminus P\}$.
\end{example}
This way of describing eventually periodic sets has several advantages:
the parameter has only three components,
the set~$I$ is uniquely defined
and it allows to determine if~$S$ is purely periodic (Lemmas~\ref{l.exis-uniq-I} and~\ref{l.ip<->I-empt}),

\begin{lemma}\label{l.exis-uniq-I}
  Let~$S$ be an eventually periodic set of integers.
  There is a unique purely periodic set~$P\subseteq\N$ and a unique finite set~$I\subseteq \N$ of mismatches
  such that~$S=I\oplus P$.
\end{lemma}

We then say that the triplet~$(p,R,I)$ is \emph{the proper parameter} of an eventually
periodic set~$S$ if~$S = (R+p\N) \oplus I$ and~$p$ is the smallest positive period
for which such~$R,I$ exist.
We take the convention that the proper parameter of a finite set~$S$ is~$p,R,I=(1,\emptyset, S)$, (instead of considering that the period equals~$0$); this is why the smallest period is assumed to be positive in the previous sentence.

\begin{lemma}\label{l.ip<->I-empt}
  An eventually periodic set~$S$ of parameter~$(p,R,I)$ is purely periodic if and
  only if~$I$ is empty.
\end{lemma}
\begin{notation}
  In what follows, we consider impurely periodic sets of integers,
  hence a finite \textbf{non-empty} set~$I\subseteq\N$ of mismatches is given.
  Moreover, we still follow the convention of Notation~\ref{n.nota} recapped hereafter.
  A period~$p$ and  a remainder-set~${R\subseteq\{0,\ldots,p-1\}}$ are given.
  We let~$k$ denote the greatest
  divisor of~$p$ that is coprime with the base~$b$,~$d$ is the integer such that~$kd=p$
  and~$j$ is the smallest integer such that~$d$ divides~$b^j$.
\end{notation}

\subsection{The automata~$\BIb$ and~$\CRpIb=\ARpb\oplus\BIb$}

We will describe an automaton accepting, by value, an eventually periodic set with parameter $(p,R,I)$. We first have to deal with the set $I$ of mismatches.
\begin{definition}
  We denote by~$m$ the greatest element of~$I$.
  We denote by~$\BIb$ the automaton:
  \begin{equation*}
      \BIb=\langle \intalph ,~
             \{ 0,\ldots,m\}\cup\{\bot\} ,~
             \delta ,~
             0 ,~
             I  \rangle~,
  \end{equation*}
  where~$\delta$ if defined as follows.
  \begin{gather}
    \forall i\in\{ 0,\ldots,m\},~\forall a\in\intalph \quad
      \left\lbrace\begin{array}{ll}
        i\rpath{a} (ib+a) \quad &\text{if  }ib+a\leq m \\
        i\rpath{a} \bot \quad &\text{otherwise}
      \end{array}\right. \\
    \forall a\in\intalph \quad \bot \rpath{a} \bot
  \end{gather}

\end{definition}
Simple and formal verification yields the following properties of~$\BIb$.
We write \emph{scc} for strongly connected component. A \emph{trivial} scc is a
state belonging to no circuit.
\begin{lemma}\label{l.BIb-prop}
  The automaton~$\BIb$
  \begin{enumerate}
    \renewcommand{\theenumi}{\alph{enumi}}
    \item is deterministic, complete and trim;
    \item has exactly two non-trivial sccs:~$\{0\}$ and~$\{\bot\}$;
    \item has exactly two 0-circuits, the respective self-loops on~$0$ and~$\bot$;
    \item accepts a word~$u\in\intalph^*$ if and only if~$\val{u}\in I$.
  \end{enumerate}
\end{lemma}

In the next definition, the exclusive disjunction $\oplus$ is extended to sets of pairs of states.
\begin{definition}\label{d.excl-disj}
  Given two complete automata~$\Ac$ and~$\Bc$ over the alphabet~$\intalph$.
  We define the \emph{exclusive disjunction}~$\Ac\oplus\Bc$ as usual:
  \begin{equation*}
      \Ac \oplus \Bc = \langle \intalph ,~
                              Q_\Ac\times Q_\Bc ,~
                              \delta ,~
                              (i_\Ac,i_\Bc),~
                              (F_\Ac \times Q_\Bc) \oplus (Q_\Ac \times F_\Bc)
                              \rangle~,
  \end{equation*}
  where~$\delta$ is defined as follows.
  \begin{equation*}
      \forall s,s'\in Q_\Ac,~\forall t,t'\in Q_\Bc,~\forall a\in\intalph\quad
  (s,t)\rpath{a} (s',t') \iff \left\lbrace\begin{array}{l}
                                  s \rpath{a} s' \text{ in } \Ac \\
                                  t\rpath{a} t' \text{ in } \Bc\\
                                  \end{array}\right.~.
  \end{equation*}
\end{definition}
It is quite obvious that a word~$u$ is accepted by~$\Ac\oplus\Bc$ if and only if
it is accepted by~$\Ac$ or~$\Bc$, but not by both of them.
\begin{notation}
  We let~$\CRpIb$ denote the exclusive disjunction~$\CRpIb=\ARpb\oplus\BIb$.
\end{notation}

The next lemma gives properties of $\CRpIb=\ARpb\oplus\BIb$ that follow
from Lemma~\ref{l.BIb-prop}
and Definition~\ref{d.excl-disj}.
Recall that we have seen in Property~\ref{pp.sc} that $\ARpb$ is strongly connected.

\begin{lemma}\label{l.CRpIb-props}
  The following properties hold.
  \begin{enumerate}
    \renewcommand{\theenumi}{\alph{enumi}}
    \item $\CRpIb$ is deterministic, complete and trim.
    \item \sublabel{l.CRpIb-sccs}
      $\CRpIb$ possesses exactly two non-trivial sccs:%
      \vspace{-0.5\topsep}
      \begin{itemize}%
        \itemsep0em
        \item the singleton made of the initial state,~$\{(0,0)\}$,
        \item and~$\{\,(s,\bot)~|~s \text{ is a state of }\ARpb\,\}$.
      \end{itemize}
      \vspace{-0.5\topsep}
      Moreover this second scc is isomorphic to~$\ARpb$ by projecting to the first component,
      hence complete.
    \item $\CRpIb$  accepts a word~$u\in\intalph^*$ if and only if~$\val{u}\in ((R+p\N)\oplus I)$.
  \end{enumerate}
\end{lemma}


%

\subsection{The~$0$-circuits of~$\CRpIb$ and of its minimisation}

The next statement gives a description of the~$0$-circuits of~$\CRpIb$ and follows from
Lemmas~\ref{l.0-circ} and~\ref{l.CRpIb-props}.
\begin{lemma}\label{l.0-circ-bis}
  A state~$s$ belongs to a~$0$-circuit of~$\CRpIb$ if and only if either
  \begin{enumerate}
    \renewcommand{\theenumi}{\alph{enumi}}
    \item the state $s$ is initial, or
    \item there exists~$x\in\ZZ[k]$ such that~$s=(x d, \bot)$.
  \end{enumerate}
\end{lemma}
The relationship of~$\CRpIb$ with its minimisation is stated by the next lemma. It is similar to the one of~$\ARpb$ with its minimisation.
\begin{lemma}\label{l.not-nero-equi-bis}
  If~$(p,R,I)$ is proper, the following statements hold.
  \begin{enumerate}
    \renewcommand{\theenumi}{\alph{enumi}}

    \item
      For every distinct integers~$x,x'\in\ZZ[k]$, the
      states~$(x d,\bot)$ and~$(x' d,\bot)$ of~$\CRpIb$ are not Nerode-equivalent.

    \item
      The states of~$\CRpIb$ that belong to~$0$-circuits are pairwise
      Nerode-inequivalent.

    \item  \sublabel{l.CRpIb-init-not-Nero-equi}
      The initial state of~$\CRpIb$ is not Nerode-equivalent
      to any other state.
  \end{enumerate}
\end{lemma}
\begin{proof}
    Item~$(a)$ follows directly from Lemmas~\ref{l.not-nero-equi} and~\ref{l.CRpIb-sccs}.

    \medskip

    $(b)$.
    From item~$(a)$ and Lemma~\ref{l.0-circ-bis}, it suffices to show that
    there is no state~$s=(x d, \bot)$ which is Nerode-equivalent to the
    initial state.
    For the sake of contradiction let us assume that such a state exists.
    We denote by~$i$ the initial state of~$\CRpIb$.
    Let~$v$ be any word whose run reaches~$s$, hence satisfying~$i\cdot v=s$.
    Since~$i$ bears a loop labelled by~$0$, the run of~$0v$ reaches~$s$,  hence, without
    loss of generality, we may
    assume that~$\wlen{v}\equiv 0~[\psi]$.
    Since~$s$ and~$i$ are Nerode-equivalent, so are~$s\cdot v$
    and~$(i\cdot v)=s$.
    By iterating this reasoning, we obtain that~$i$ is Nerode-equivalent
    to~$(i\cdot v^k)$.
    Similarly,~$i$ is Nerode-equivalent
    to the state~$(i\cdot v^k 0^j)$, that we denote by~$s'$.
    Moreover, since~$s'$ is reachable from~$s$ and since~$s$ belongs to the~$\bot$-scc
    (complete from Lemma~\ref{l.CRpIb-sccs}),~$s'$ belongs to the~$\bot$-scc as well.
    Since~$\wlen{v}\equiv 0~[\psi]$, $\val{v^k}\equiv k\val{v}\equiv0~[k]$, hence $\val{v^k 0^j}\equiv0~[k]~$. Since~$\val{v^k 0^j}$ is obviously a multiple of~$d$, it is also a multiple of~$p=kd$.
    In other, words~$s'=(0, \bot)$ and the initial state is Nerode-equivalent to~$(0, \bot)$.
    This contradicts the fact that~$I$ is non-empty.

    \medskip

    $(c)$.
    Let us denote by~$X$ the Nerode-equivalence class of the initial state.
    Since the initial state bears a loop labeled by~$0$, the set~$X$ is stable by
    reading the digit~$0$.
    Therefore, if~$X$ were containing a non-initial state, then it would contain a whole
    $0$-circuit (distinct from the initial state), contradicting item~$(b)$.
\end{proof}
The next statement can then be established using Lemma~\ref{l.not-nero-equi-bis}$(b)$
much like Proposition~\ref{p.not-nero-equi} was shown using Lemma \ref{l.not-nero-equi}.
\begin{proposition}\label{p.not-nero-equi-bis}
  If~$(p,R,I)$ is proper, then two Nerode-equivalent
  states~$(s',t')$ and~$(s',t')$ are necessarily such that~$s$ and~$s'$ are congruent
  modulo~$k$.
\end{proposition}
%


%

%
It follows from Lemma~\ref{l.0-circ-bis} that~$\CRpIb$ has~$(k+1)$
states that belong to~$0$-circuits and from
Lemma~\ref{l.not-nero-equi-bis}(b) that such states are not merged by the
minimisation process, hence the next proposition holds.

\begin{proposition}\label{p.mini-k-stat-bis}
  If~$(p,R,I)$ is proper, the 0-circuits of the minimisation of~$\CRpIb$
  have a total of~$(k+1)$ states
\end{proposition}
%

\subsection{Ultimate equivalence class of~$\CRpIb$}

\begin{lemma}\label{l.equiv-k-ulti-equiv-bis}
  Let~$(s,t)$ and~$(s',t')$ be two states of~$\CRpIb$ such that~$s\equiv s'[k]$.
  If neither~$(s,t)$ nor~$(s',t')$ is the initial state,
  then~$(s,t)$ and~$(s',t')$ are ultimately-equivalent.
\end{lemma}
\begin{proof}
  Since by hypothesis,~$(s,t)$ and~$(s',t')$ are not initial, there exists a bound~$m$
  such that for every word~$u$ longer that~$m$,
  the states~$(s,t)\cdot u$ and~$(s,t')\cdot u$ belong to the~$\bot$-scc.
  Without loss of generality, we may assume that~$m\geq j$.
  Let~$u$ be a word longer than~$m$.
  Then,
  \begin{equation*}
      (s,t) \cdot u = (sb^{\wlen{u}}+\val{u},\bot) \text{ and } (s',t') \cdot u = (s'b^{\wlen{u}}+\val{u},\bot)~.
  \end{equation*}
  Since~$s$ and~$s'$ are congruent modulo~$k$, and since~$k$ is coprime with b, it holds:
  \begin{equation}\label{eq.equiv-k-ulti-equiv-1}
    sb^{\wlen{u}}+\val{u}\equiv s'b^{\wlen{u}}+\val{u}~[k]~.
  \end{equation}
  Moreover, since~$\wlen{u} \geq j$,~ and since~$d$ divises~$b^j$,
  \begin{equation}\label{eq.equiv-k-ulti-equiv-2}
    sb^{\wlen{u}}+\val{u}\equiv \val{u} \equiv s'b^{\wlen{u}}+\val{u}~[d]~.
  \end{equation}
  Finally, since~$d$ and~$k$ are coprime, (\ref{eq.equiv-k-ulti-equiv-1}) and (\ref{eq.equiv-k-ulti-equiv-2}) yield
  \begin{equation*}
      sb^{\wlen{u}}+\val{u}\equiv s'b^{\wlen{u}}+\val{u}~[p]~,
  \end{equation*}
  hence~$(s,t) \cdot u = (s',t') \cdot u$.
\end{proof}

\begin{lemma}
  The initial state of~$\CRpIb$ is not ultimately equivalent to any other state.
\end{lemma}
\begin{proof}
  The only state from whom the initial state may be reached is the initial state itself.
  Moreover, as the initial state bears a loop labelled by~$0$, the words of~$0^*$
  are witnesses of the fact that no state is ultimately equivalent to the initial state.
%
%
%
\end{proof}

\section{Characterisation of automata accepting impurely periodic sets}

\begin{theorem}\label{t.ip}
  Let~$b>1$ be a base and~$\Ac$ be a minimal automaton over~$\intalph$.
  We write~$(\ell+1)$ for the number of states in~$\Ac$ that belong to~$0$-circuits.
  The automaton~$\Ac$ accepts by value an \textbf{im}purely periodic set of integers if and only if
  the following conditions are met.
  \begin{enumerate}
    \renewcommand{\theenumi}{\alph{enumi}}
    \item \sublabel{t.ip-pseu-morp}
      There exists a pseudo-morphism~$\phi:\Ac\rightarrow \ARpb[?][\ell]$.
    \item \sublabel{t.ip-ulti-equi}
      The initial state excluded, the equivalence relation induced by~$\phi$ is a refinement
      of the ultimate-equivalence relation.
    \item
      The initial state bears a self-loop labelled by the digit~$0$ and features no
      other incoming transitions.
    \sublabel{t.ip-init}
  \end{enumerate}
\end{theorem}
\begin{proof}[Proof of forward direction]
  Conditions \ref{t.ip-pseu-morp*} and \ref{t.ip-ulti-equi*} are obtained
  much like it was done in Theorem~\ref{t.pp}.
  We simply apply Propositions~\ref{p.mini-k-stat-bis} and \ref{p.not-nero-equi-bis}
  instead of~\ref{p.mini-k-stat} and \ref{p.not-nero-equi}.
  Since~$\Ac$ is minimal and accepts by value an impurely periodic set, there exists a parameter~$(p,R,I)$
  such that~$\Ac$ is the minimisation of~$\CRpIb$.
  A simple verification yields that Condition~\ref{t.ip-init*} is satisfied by~$\CRpIb$.
  Besides, it follows from Lemma~\ref{l.CRpIb-init-not-Nero-equi} that
  the minimisation process does not merge any state of~$\CRpIb$ with the initial state.
  As a result, the incoming transitions to the initial state are the same in~$\Ac$ and~$\CRpIb$.
\end{proof}

\begin{proof}[Proof of backward direction]
  There are finitely many ultimate-equivalence classes.
Hence there exists an integer $m$ such that, if two states $s$ and $s'$ are ultimately equivalent, then they are $m$-ultimately-equivalent.

  Note also that since~$\Ac$ is complete, Condition~\ref{t.ip-init*} implies that~$\ell\geq 1$.
  Let~$u,u'$ be two words whose respective values are congruent modulo~$\ell b^m$
  and greater than~$b^m$.
  Thus, there are words~$v,v',w,w'$, $\wlen{w}=\wlen{w'}=m$, satisfying~$u=vw$,~$u'=v'w'$
  and such that~$v,v'$ both possess a non-zero digit.
  In particular, neither~$\Ac\cdot v$ nor~$\Ac\cdot v'$ is the initial state.
  With exactly the same proof as was given in Theorem~\ref{t.pp}, it may then be shown
  that~$\Ac\cdot u = \Ac \cdot u'$.
  In other words,~$\Ac$ accepts an ultimately periodic set of integers~$S$ of period~$\ell b^m$.
  (In general, this period is not the smallest one, which would be~$\ell d$
  for some~$d$ dividing~$b^m$.)
  We moreover write~$I$ the set of mismatches (existence and unicity ensured by
  Lemma~\ref{l.exis-uniq-I}).
  Let us show that it is not purely periodic, or equivalently that~$I$ is not empty
  (Lemma~\ref{l.ip<->I-empt}).
  We denote by~$s$ the state reached by the run of the word~$\rep{\ell b^m}$, \emph{i.e.},~$s=\Ac\cdot \rep{\ell b^m}$.
  Since~$\ell \geq 1$, this word possesses a non-zero digit, hence~$s$ is not the initial state
  of~$\Ac$ (Condition \ref{t.ip-init*}).
  Since~$\Ac$ is minimal,~$s$ and~$i_\Ac$ are not Nerode-equivalent. Hence there exists
  a word~$w$ such that exactly one of the states in~$\{s\cdot w,~i_\Ac\cdot w\}$ is final.
  Since~$\val{w}$ and~$\val{\rep{\ell b^m}w}$ are obviously congruent modulo~$\ell b^m$,~$\val{w}$
  is a mismatch: it belongs to~$I$.
\end{proof}

As stated below, Theorem~\ref{t.ip} gives an algorithm to decide whether an automaton accepts an ultimately
periodic set of integers.
It is the same as the one from Section~\ref{s.comp-algo} with an
additional Step (5) at the end.
It consists in verifying that
Condition~\ref{t.ip-init} holds.

\begin{corollary}\label{c.ip}
  Let~$b$ be a base and~$\Ac$ be a $n$-state deterministic automaton over~$\intalph$.
  It is decidable in~$O(b n \log n)$ time whether~$\Ac$ accepts by value an impurely periodic set of integers.
\end{corollary}

Since an eventually periodic set is either purely or impurely periodic, Theorem
\ref{t.main} is a direct consequence of Corollaries~\ref{c.pp} and~\ref{c.ip}.
%

  \small
  \bibliographystyle{plainurl}
  \bibliography{bibliography}

\begin{thebibliography}{10}

\bibitem{AlloucheRampersadShallit2009}
Jean-Paul Allouche, Narad Rampersad, and Jeffrey Shallit.
\newblock Periodicity, repetitions, and orbits of an automatic sequence.
\newblock {\em Theoret. Comput. Sci}, 410:2795--2803, 2009.

\bibitem{AlloucheShallit2003}
Jean-Paul Allouche and Jeffrey Shallit.
\newblock {\em {Automatic Sequences: Theory, Applications, Generalizations}}.
\newblock Cambridge University Press, 2003.

\bibitem{Beal&Crochemore2007}
Marie-Pierre B{\'e}al and Maxime Crochemore.
\newblock Minimizing local automata.
\newblock In M.~Fossorier G.~Caire, editor, {\em IEEE International Symposium
  on Information Theory}, pages 1376--1380, 2007.

\bibitem{BellCharlierFraenkelRigo2009}
Jason Bell, Emilie Charlier, Aviezri~S. Fraenkel, and Michel Rigo.
\newblock A decision problem for ultimately periodic sets in nonstandard
  numeration systems.
\newblock {\em {IJAC}}, 19(6):809--839, 2009.

\bibitem{BertRigo10-b}
Val{\'{e}}rie Berth{\'{e}} and Michel Rigo, editors.
\newblock {\em Combinatorics, Automata and Number Theory}.
\newblock Number 135 in Encyclopedia Math. Appl. Cambridge University Press,
  2010.

\bibitem{BoigelotBrusten2009}
Bernard Boigelot and Julien Brusten.
\newblock A generalization of {C}obham's theorem to automata over real numbers.
\newblock {\em Theor. Comput. Sci.}, 410(18):1694--1703, 2009.

\bibitem{BoigelotJodogneWolper2005}
Bernard Boigelot, S\'ebastien Jodogne, and Pierre Wolper.
\newblock An effective decision procedure for linear arithmetic over the
  integers and reals.
\newblock {\em ACM Trans. Comput. Log.}, 6(3):614--633, 2005.

\bibitem{BruyereHansel1995}
V.~Bruy{\`e}re and G.~Hansel.
\newblock Recognizable sets of numbers in nonstandard bases.
\newblock In R.~Baeza-Yates, E.~Goles, and P.~V. Poblete, editors, {\em LATIN
  '95: Theoretical Informatics}, volume 911 of {\em Lect. Notes Comput. Sci.},
  pages 167--179. Springer, 1995.

\bibitem{BruyereHanselMichauxVillemaire}
V.~Bruy{\`e}re, G.~Hansel, C.~Michaux, and R.~Villemaire.
\newblock Logic and $p$-recognizable sets of integers.
\newblock {\em Bull. Belg. Math. Soc.}, 1:191--238, 1994.
\newblock Corrigendum, {\it Bull.\ Belg.\ Math.\ Soc.} {\bf 1} (1994), 577.

\bibitem{CharlierRampersadShallit2012}
Emilie Charlier, Narad Rampersad, and Jeffrey Shallit.
\newblock Enumeration and decidable properties of automatic sequences.
\newblock {\em Int. J. Found. Comput. Sci.}, 23(5):1035--1066, 2012.

\bibitem{Cobham1969}
Alan Cobham.
\newblock On the base-dependence of sets of numbers recognizable by finite
  automata.
\newblock {\em Mathematical Systems Theory}, 3(2):186--192, 1969.

\bibitem{Durand2013}
Fabien Durand.
\newblock Decidability of the {HD0L} ultimate periodicity problem.
\newblock {\em {RAIRO} - Theor. Inf. and Applic.}, 47(2):201--214, 2013.

\bibitem{Honkala1986}
Juha Honkala.
\newblock A decision method for the recognizability of sets defined by number
  systems.
\newblock {\em ITA}, 20(4):395--403, 1986.

\bibitem{Leroux2005}
J{\'e}r{\^o}me Leroux.
\newblock A polynomial time {P}resburger criterion and synthesis for number
  decision diagrams.
\newblock In {\em Logic in Computer Science 2005 (LICS 2005)}, pages 147--156.
  IEEE Comp. Soc. Press, 2005.

\bibitem{MarsaultSakarovitch2013}
Victor Marsault and Jacques Sakarovitch.
\newblock {U}ltimate {P}eriodicity of b-{R}ecognisable {S}ets: {A}
  {Q}uasilinear {P}rocedure.
\newblock In Marie{-}Pierre B{\'{e}}al and Olivier Carton, editors, {\em
  Developments in Language Theory - 17th International Conference ({DLT}
  2013)}, number 7907 in Lect. Notes Comput. Sci., pages 362--373. Springer,
  2013.

\bibitem{Mitrofanov2011}
Ivan Mitrofanov.
\newblock A proof for the decidability of {HD0L} ultimate periodicity (in
  russian).
\newblock Preprint {arXiv:1110.4780}, 2011.

\bibitem{Muchnik2003}
Andrei~A. Muchnik.
\newblock The definable criterion for definability in {P}resburger arithmetic
  and its applications.
\newblock {\em Theor. Comput. Sci.}, 290(3):1433--1444, 2003.
\newblock English translation of a prior article with the same name in Russian,
  {M}oscow's {I}nstitute of New Technologies, 1991.

\bibitem{Sakarovitch2009}
Jacques Sakarovitch.
\newblock {\em Elements of Automata Theory}.
\newblock Cambridge University Press, 2009.
\newblock Corrected English translation of \textit{{\'E}l{\'e}ments de
  th{\'e}orie des automates}, {V}uibert, 2003.

\bibitem{Semenov1977}
Alexei~L. Semenov.
\newblock Presburgerness of predicates regular in two number systems.
\newblock {\em Siberian Mathematical Journal}, 18(2):289--300, 1977.
\newblock English translation from Russian Translated from Sibirskii
  Matematicheskii Zhurnal, 18(2), pp. 403–418, 1977.

\bibitem{Tarjan1972}
Robert~E. Tarjan.
\newblock Depth-first search and linear graph algorithms.
\newblock {\em {SIAM} J. Comput.}, 1(2):146--160, 1972.

\end{thebibliography}
\end{document}